\documentclass[pra,floatfix,amsmath,twocolumn]{revtex4}
\usepackage{amssymb}
\usepackage{graphicx}
\usepackage{graphics}
\usepackage{amsmath}
\usepackage{amsthm}

\def\i{{\bf i}}

\def\j{{\bf j}}

\newcommand{\bea}{\begin{eqnarray}}
\newcommand{\eea}{\end{eqnarray}}
\def\bi{\begin{itemize}}
\def\ei{\end{itemize}}
\def\bc{\begin{center}}
\def\ec{\end{center}}

\def\C{\hbox{$\mit I$\kern-.7em$\mit C$}}
\def\R{\hbox{$\mit I$\kern-.6em$\mit R$}}

\def\ket#1{|#1\rangle}
\newcommand{\one}{\mbox{$1 \hspace{-1.0mm}  {\bf l}$}}
\def\tr{\mathrm{tr}}
\def\ket#1{\left| #1\right>}
\def\bra#1{\left< #1\right|}

\newcommand{\proj}[1]{\ket{#1}\bra{#1}}

\newtheorem{theorem}{Theorem}

\newtheorem{lemma}[theorem]{Lemma}

\begin{document}

\title{Local unitary equivalence of multipartite pure states}

\author{B. Kraus}

\affiliation{Institute for Theoretical Physics, University of
Innsbruck, Austria}

\begin{abstract}

Necessary and sufficient conditions for the equivalence of
arbitrary $n$--qubit pure quantum states under Local Unitary (LU)
operations are derived. First, an easily computable standard form
for multipartite states is introduced. Two generic states are shown
to be LU--equivalent iff their standard forms coincide. The
LU--equivalence problem for non--generic states is solved by
presenting a systematic method to determine the LU operators (if
they exist) which interconvert the two states.
\end{abstract}
\maketitle

Multipartite states occur in many applications of quantum
information, like one--way quantum computing, quantum error
correction, and quantum secret sharing \cite{Gothesis97,RaBr01}.
Furthermore, the theory of multipartite states plays also an
important role in other fields of physics which deal with many-body
systems \cite{AmFa08}. The existence of those practical and
abstract applications is due to the subtle properties of
multipartite entangled states. Thus, one of the main goals in
quantum information theory is to gain a better understanding of the
non--local properties of quantum states. Whereas the bipartite case
is well understood, the multipartite case is much more complex.
Even though a big theoretical effort has been undertaken where
several entanglement measures for multipartite states have been
introduced \cite{CoKu00}, different classes of entangled states
have been identified \cite{DuViCi00}, and a normal form of
multipartite states has been presented \cite{Ves}, we are far from
completely understanding the non--local properties of multipartite
states \cite{HoHo07}.

One approach to gain insight into the entanglement properties of
quantum states is to consider their interconvertability. That is,
given two states $\ket{\Psi}$, $\ket{\Phi}$ the question is whether
or not $\ket{\Psi}$ can be transformed into $\ket{\Phi}$ by local
operations \cite{HoHo07}. One particularly interesting case, which
is also investigated in this paper, is the LU-equivalence of
multipartite states. We say that a $n$--partite state, $\ket{\Psi}$
is LU--equivalent to $\ket{\Phi}$ ($\ket{\Psi}\simeq_{LU}
\ket{\Phi}$) if there exist local unitary operators, $U_1,\ldots,
U_n$, such that $\ket{\Psi}=U_1\otimes \cdots \otimes U_n
\ket{\Phi}$. Note that two states which are LU--equivalent are
equally useful for any kind of application and they posses
precisely the same amount of entanglement. This is why
understanding the interconvertability of quantum states by LU
operations is part of the solution to the more general problem of
characterizing the different types of entangled quantum states.

In order to solve this long--standing problem the so--called local
polynomial invariants have been introduced \cite{GrRo98}. However,
even though it is known that it is sufficient to consider only a
finite set of them, this complete finite set is known only for very
few simple cases.

Here, we derive necessary and sufficient conditions for the
existence of LU operations which transform two states into each
other. For generic states, states where non of the single qubit
reduced states is completely mixed, the conditions can be easily
computed. For arbitrary $n$--qubit states a systematic method to
determine the unitaries (in case they exist) which interconvert the
states is presented.

The sequel of the paper is organized as follows. First, we
introduce a standard form of multipartite states, which we use in
order to derive easily computable necessary and sufficient
conditions for the LU--equivalence of generic multipartite states.
Like in the bipartite case, it is shown that two generic states are
LU--equivalent iff their standard forms coincide. For non--generic
states it is shown that whenever one of the single qubit reduced
states is not completely mixed, the problem of LU--equivalence of
$n$--qubit states can be reduced to the problem of LU--equivalence
of $(n-1)$--qubit states. Then, a systematic method to determine
the local unitaries (if they exist) which interconvert two
arbitrary states is presented. It is shown that the states are
LU--equivalent iff there exists a solution to a finite set of
equations. The number of variables involved in those equations
depends on the entanglement properties of the states. The case with
the largest number of variables occurs for the sometimes called
maximally entangled states of $n$ qubits, where any bipartition of
$\lceil n/2 \rceil$ qubits is maximally entangled with the rest. It
is known however, that only for certain values of $n$ such states
exist \cite{Bra03}. The power of this method is illustrated by
considering several examples.

Throughout this paper the following notation is used. By $X,Y,Z$ we
denote the Pauli operators. The subscript of an operator will
always denote the system it is acting on, or the system it is
describing. The reduced states of system $i_1,\ldots i_k$ of
$\ket{\Psi}$ ($\ket{\Phi}$) will always be denoted by
$\rho_{i_1\ldots i_k}$ ($\sigma_{i_1\ldots i_k}$) resp., i.e.
$\rho_{i_1\ldots i_k}=\tr_{\neg i_1\ldots \neg i_k}(\proj{\Psi})$.
We denote by ${\bf i}$ the classical bit--string $(i_1,\ldots,
i_n)$ with $i_k\in\{0,1\}$ $\forall k\in \{1,\ldots, n\}$ and
$\ket{{\bf i}}\equiv\ket{i_1,\ldots, i_n}$ denotes the
computational basis. Normalization factors as well as the tensor
product symbol will be omitted whenever it does not cause any
confusion.

Let us start by introducing a unique standard form of multipartite
states (see also \cite{KrKr08}). Let $\ket{\Psi}$ be a $n$--qubit
state. As a first step we apply local unitaries, $U_i^1$ such that
all the single qubit reduced states of the state
$\ket{\Psi_t}=U_1^1\otimes \ldots U_n^1 \ket{\Psi}$ are diagonal in
the computational basis, i.e. $\tr_{\neg
i}(\proj{\Psi_t})=D_i=\mbox{diag}(\lambda_i^1,\lambda_i^2)$. We
call any such decomposition trace decomposition of the state
$\ket{\Psi}$. A sorted trace decomposition is then defined as a
trace decomposition with $\lambda_i^1\geq \lambda_i^2$. Note that
transforming a state into its sorted trace decomposition, which we
will denote by $\ket{\Psi_{st}}$ in the following, can be easily
done by computing the spectral decomposition of all the single
qubit reduced states. The sorted trace decomposition of a generic
state, $\ket{\Psi}$ with $\rho_i\neq \one$ $\forall i$ is unique up
to local phase gates. That is $U_1\ldots U_n \ket{\Psi_{st}}$ is a
sorted trace decomposition of $\ket{\Psi}$ iff (up to a global
phase, $\alpha_0$) $U_i= U_i(\alpha_i)\equiv \text{
diag}(1,e^{i\alpha_i })$. In order to make the sorted trace
decomposition of generic states unique we impose the following
condition on the phases $\alpha_i$, $i\in\{0,\ldots, n\}$. We write
$\ket{\Psi_{st}}=\sum_{i_1,\ldots i_n=0}^1
\lambda_{i_1,\ldots,i_n}\ket{i_1,\ldots,i_n}$, and define the set
$S=\{{\bf i}: \lambda_{{\bf i}}\neq 0\}$ and $\bar{S}$ denotes the
set of the linearly independent vectors in $S$. The global phase,
$\alpha_0$ is chosen to make $\lambda_{{\bf i_0}}$ real and
positive where ${\bf i_0}={\bf 0}$ in case $\lambda_{{\bf 0}}\neq
0$ else ${\bf i_0}$ denotes the first (in lexicographic order)
linearly dependent vector in $S$. After that, the $n$ phases are
chosen to make the coefficients $e^{i\alpha_0}\lambda_{{\bf i}}$
for ${\bf i}\in \bar{S}$ real and positive \footnote{If there are
less than $n$ linearly independent vectors in $S$, say $k$, then
$k$ phases can be defined like that, the other phases leave the
state invariant and can therefore be chosen arbitrarily.}. Since
all the phase gates, which do not leave the state invariant are
fixed in this way we have that $U_1\ldots U_n\ket{\Psi_s}$, where
$\ket{\Psi_s}$ denotes here and in the following the standard form
of $\ket{\Psi}$, has standard form iff $U_1\ldots
U_n\ket{\Psi_s}=\ket{\Psi_s}$. That is the standard form is unique.
If $\rho_i=\frac{1}{2}\one$, for some system $i$, the standard form
can be similarly defined \cite{KrKr08}, however it will not be
unique then. Due to the definition any state is LU--equivalent to
its standard form \footnote{Note that this standard form coincides
for the simplest case of two qubits with the Schmidt decomposition
\cite{NiCh00} and can be generalized to $d$--level systems.}.

We employ now the standard form to derive a criterion for the
LU--equivalence of generic multipartite states. First of all note
that $\ket{\Psi}\simeq_{LU}\ket{\Phi}$ iff
$\ket{\Psi_s}\simeq_{LU}\ket{\Phi_s}$. Using then that the standard
form is unique we obtain the following theorem.

\begin{theorem} Let $\ket{\Psi}$ be an $n$ qubit state with $\rho_i\neq \one$ $\forall i$. Then $\ket{\Psi}\simeq_{LU} \ket{\Phi}$
iff the standard form of $\ket{\Psi}$ is equivalent to the standard
form of $\ket{\Phi}$, i.e. $\ket{\Psi_s}=\ket{\Phi_s}$.
\end{theorem}

Thus, similarly to the bipartite case, two generic states are
LU--equivalent iff their standard forms coincide, which can be
easily checked. Furthermore, if the states are LU--equivalent then
$\ket{\Psi}=U_1,\ldots,U_n\ket{\Phi}$ with $U_i=(U_s^i)^\dagger
V_s^i$, where $U_s^i$, $V_s^i$ denote the local unitaries such that
$\ket{\Psi_s}\equiv U_s^1\ldots U_s^n\ket{\Psi}$ and
$\ket{\Phi_s}\equiv V_s^1 \ldots V_s^n\ket{\Phi}$.

In order to study now the non--generic cases, we will rewrite the
necessary and sufficient condition derived above. For a generic
state, $\ket{\Psi}$ it is easy to verify that
$\ket{\Psi_s}=\ket{\Phi_s}$ iff there exists a bitstring ${\bf
k}=k_1,\ldots k_n$, local phase gates $U_i(\alpha_i)$, and a global
phase $\alpha_0$ s.t. \bea \label{LU} e^{i\alpha_0} \bigotimes_i
 U_i(\alpha_i)X_i^{k_i}\bar{W}_i\ket{\Psi}=\bigotimes_i \bar{V}_i
\ket{\Phi},\eea where $\bar{W}_i$ ($\bar{V}_i$) are local unitaries
which transform $\rho_i$ ($\sigma_i$) into a diagonal matrix. That
is $\bigotimes_i\bar{W}_i\ket{\Psi}$ and $ \bigotimes_i\bar{V}_i
\ket{\Phi}$ are trace decompositions of $\ket{\Psi}$ and
$\ket{\Phi}$ resp.. For generic states $k_i$ is chosen such that
the order of the eigenvalues of the single qubit reduced states of
$\bigotimes_i
 X_i^{k_i}\bar{W}_i\ket{\Psi}$ and $\bigotimes_i \bar{V}_i
\ket{\Phi}$ coincides. In order to check then whether or not there
exist phases $\alpha_i$ such that Eq. (\ref{LU}) is satisfied, we
make use of the following lemma. There, we will consider four $n$--
qubit states. The systems, each composed out of $n$ qubits will be
denoted by $A,B,C,D$ respectively. The $i$-th qubit of system $A$
will be denoted by $A_i$, etc. Furthermore, we will use the
notation $\ket{\chi_i}=(\ket{0110}-\ket{1001})_{A_i,B_i,C_i,D_i}$
and $P^i_{AC}=\sum_{\bf k} \ket{{\bf k}}\bra{{\bf k}{\bf
k}}_{A_1,C_1,\ldots A_{i-1},C_{i-1},A_{i+1},C_{i+1}\ldots,
A_n,C_n}$ and similarly we define $P^i_{BD}$ for systems $B,D$. For
a state $\ket{\Psi}$ we define $K_\Psi\equiv \{{\bf k} \mbox{ such
that }\bra{{\bf k}}\Psi\rangle=0\}$ and
$\ket{\Psi_{\bar{\alpha}_i}}=\ket{\Psi}+ e^{-i\bar{\alpha}_0}
\sum_{{\bf k}\in K_\Psi} e^{-i\sum_{i=1}^n \bar{\alpha}_i
k_i}\ket{{\bf k}}$ for some phases $\bar{\alpha}_i$ and
$\ket{\Psi_{\bf 0}}=\ket{\Psi}+\sum_{{\bf k}\in K_\Psi} \ket{{\bf
k}}$.

\begin{lemma}
\label{LemmaPhase} Let $\ket{\Psi}, \ket{\Phi}$ be $n$ qubit
states. Then, there exist local phase gates, $U_i(\alpha_i)$ and a
phase $\alpha_0$ such that
$\ket{\Psi}=e^{i\alpha_0}\bigotimes_{i=1}^n U_i(\alpha_i)
\ket{\Phi}$ iff there exist phases $\{\bar{\alpha}_i\}_{i=0}^{n}$
such that \bi \item[(i)] $|\bra {\bf i} \Psi_{\bf 0}\rangle|= |\bra
{\bf i} \Phi_{\bar{\alpha}_i})\rangle|$ $\forall {\bf i}$ and
\item[(ii)] $\bra{\chi}_i P^i_{AC} P^i_{BD} \ket{\Psi_{\bf
0}}_A\ket{\Psi_{\bf
0}}_B\ket{\Phi_{\bar{\alpha}_i}}_C\ket{\Phi_{\bar{\alpha}_i}}_D=0$
$\forall i\in \{1,\ldots ,n\}$.\ei

\end{lemma}

The prove of this lemma will be presented in the appendix.


Let us now consider the non--generic case. Obviously, two arbitrary
states, $\ket{\Psi},\ket{\Phi}$, are LU--equivalent iff there exist
local unitaries $\bar{V}_k,\bar{W}_k$ a bit string ${\bf k}$ and
phases $\alpha_i$ such that Eq. (\ref{LU}) is fulfilled. We will
show now how $\bar{V}_k,\bar{W}_k$ can be determined by imposing
necessary conditions of LU--equivalence.

First of all, we note that for any state $\ket{\Psi}$ with
$\rho_i\neq \one$ for some system $i$, $k_i$ as well as $\bar{V}_i$
and $\bar{W}_i$ can be easily determined as follows. If
$\ket{\Psi}\simeq_{LU}\ket{\Phi}$ then all the reduced states must
be LU--equivalent, in particular
$D_i=\mbox{diag}(\lambda_1^i,\lambda_2^i)=\bar{W}_i \rho_i
\bar{W}_i^\dagger=\bar{V}_i \sigma_i \bar{V}_i^\dagger$, for some
unitaries $\bar{W}_i,\bar{V}_i$. Analogously to the generic case,
this equation determines $\bar{W}_i$ and $\bar{V}_i$ (and $k_i=0$)
uniquely up to a phase gate. Thus, for this case we have that
$\ket{\Psi}\simeq_{LU} \ket{\Phi}$ iff there exist two phases,
$\alpha_{i}$ and $\alpha_0$ and local unitaries $U_j$ such that
\bea \label{cond1} \phantom{,}_i\bra{l}\bar{W}_i\Psi_s\rangle
=e^{i(\Phi+\alpha_i l)}\bigotimes_{j\neq
i}U_{j}\phantom{,}_i\bra{l}\bar{V}_i\Phi_s\rangle, \eea where
$l\in\{0,1\}$ and $\bar{W}_i,$ and $\bar{V}_i$ are chosen such that
$D_i=\mbox{diag}(\lambda_1^i,\lambda_2^i)=\bar{W}_i \rho_i
\bar{W}_i^\dagger=\bar{V}_i \sigma_i \bar{V}_i^\dagger$. Hence, if
there is one system where the reduced state is not proportional to
the identity then we can reduce the problem of LU--equivalence of
$n$--qubit states to the LU--equivalence of $(n-1)$--qubit states.
This statement can be easily generalized to the case where more
than one single qubit reduced state is not completely mixed.


Let us now consider the more complicated case, where some
$\rho_i=\one$. There, it is obviously no longer possible to
determine $\bar{V}_i$,$\bar{W}_i$ by imposing the necessary
condition of LU--equivalence, $\rho_i=U_i\sigma_i U^\dagger_i$.
However, we will show next, which necessary condition can be used
in order to determine them. Before we do so, we explain the problem
which might occur if $\rho_i=\one$ by considering a simple example.
Let $\ket{\Psi}$ and $\ket{\Phi}$ denote two states with
$\rho_{12}=\sigma_{12}=\one-\lambda\proj{\Psi^-}$, for some
$\lambda\neq 0$. Then we find that $\rho_{12}=U_1 U_2 \sigma_{12}
U_1^\dagger U_2^\dagger$ iff $U_1=U_2$, which implies that
$\ket{\Psi}\simeq_{LU} \ket{\Phi}$ iff there exist local unitaries
$U_1,U_3,\ldots, U_n$ such that $\ket{\Psi}=U_1U_1\ldots
U_n\ket{\Phi}$. Thus, the unitary $U_2$ depends on $U_1$. Or,
stated differently, $\bar{W}_2$ (and $\alpha_2$) depends on $U_1$
in Eq. (\ref{LU}), where we set $V_1=V_2=\one$. In general we might
neither be able to determine the phase $\alpha_2$, nor $\bar{W}_2$
as a function of $U_1$ alone. However, the next lemma shows that
any $\bar{W}_k$ can be determined as a function of a few unitaries
and $\bar{V}_k$ can always be determined directly form the state
$\ket{\Phi}$. We will see that the number of unitaries which are
required to define $\bar{W}_k$ depends on the entanglement
properties of the state.

\begin{lemma} \label{Le1} If $\ket{\Psi}=U_1\ldots U_n\ket{\Phi}$ and
if there exist systems $i_1, \ldots i_l$ such that $\rho_{i_1,
\ldots i_l,k}\neq \rho_{i_1, \ldots i_l}\otimes \one$ then
$\bar{V}_k$ in Eq. (\ref{LU}) can be determined from the state
$\ket{\Phi}$ and $\bar{W}_k$ can be determined as a function of
$U_{i_1},\ldots U_{i_l}$.\end{lemma}
\begin{proof}

Without loss of generality we assume $i_1=1, \ldots i_l=l$ and
write $\ket{\Psi}=\sum \ket{\i}_{1\ldots
,l}\ket{\Psi_{\i}}_{l+1\ldots ,n}$ and $\ket{\Phi}=\sum
\ket{\i}_{1\ldots ,l}\ket{\Phi_{\i}}_{l+1\ldots ,n}$, where
$\i=(i_1,\ldots,i_l)$. Since $\sigma_{1,\ldots,l,k}= \sum \ket{\i}
\bra{\j} \tr_{\neg k}(\ket{\Phi_{\i}} \bra{\Phi_{\j}})\neq
\sigma_{1,\ldots,l}\otimes \one$, there exist at least two tuples
$\i$ and $\j=(j_1,\ldots j_l)$ such that the $2\times 2$ matrix
$X_{\i}^{\j}\equiv \tr_{\neg k}(\ket{\Phi_{\i}}
\bra{\Phi_{\j}})\not\propto \one $. Thus, at least one of the two
hermitian operators $Y_{\i}^{\j}=X_{\i}^{\j}+(X_{\i}^{\j})^\dagger$
and $Z_{\i}^\j=i X_\i^\j-i(X_\i^\j)^\dagger$ is not proportional to
the identity. W. l. o. g. we assume that $\one \not\propto Y_\i^\j
= \tr_{\neg k}[(\ket{\i}\bra{\j}+h.c)\ket{\Phi}\bra{\Phi}]$. Using
that $\ket{\Psi}=U_1\ldots U_n\ket{\Phi}$ we have \bea \label{Yi}
U_k Y_\i^\j U_k^\dagger =\tr_{\neg k}[(\ket{{\bf i}}\bra{{\bf
j}}+h.c)\cdot U_1^\dagger \ldots U_l^\dagger
\ket{\Psi}\bra{\Psi}U_1 \ldots U_l].\eea

Since $Y_\i^\j$ is hermitian we can diagonalize it as well as the
right hand side of Eq (\ref{Yi}) and obtain $U_k \bar{V}^\dagger_k
D \bar{V}_k U_k^\dagger=\nonumber
\\\bar{W}^\dagger_k(U_1,\ldots U_l)  D(U_1,\ldots
U_l)  \bar{W}_k(U_1,\ldots U_l),$ which is true iff $D= X^{i_k}D
(U_1,\ldots U_l)X^{i_k} $, with $i_k\in\{0,1\}$ and $U_k=
e^{i\alpha_0}\bar{W}^\dagger_k(U_1,\ldots U_l) U(\alpha_k) X^{i_k}
\bar{V}_k,$ for some phases $\alpha_0,\alpha_k$. Note that Eq
(\ref{Yi}) must hold for any $\i,\j$. Note further that $\bar{V}_k$
is the unitary which diagonalizes $Y_\i^\j$ and can therefore be
determined directly from the state $\ket{\Phi}$. Thus, we have
$\ket{\Psi}=U_1\ldots U_n\ket{\Phi}$ iff there exists
$i_k\in\{0,1\}$, $\alpha_0$ and $\alpha_k$ such that
$e^{i\alpha_0}X^{i_k}U(\alpha_k)\bar{W}_k(U_1,\ldots,U_l)
\ket{\Psi}=U_1\ldots \bar{V}_k \ldots U_n\ket{\Phi}$. \end{proof}

Note that the proof of Lemma \ref{Le1} is constructive. The idea
was to impose the necessary condition for LU--equivalence given in
Eq. (\ref{Yi}) for any $l$--tuples $\bf{i},\bf{j}$. Since the
$2\times 2$ matrices occurring in this equation are hermitian, one
can, similarly to the previous cases, determine the unitaries
$\bar{V}_k, \bar{W}_k$ by diagonalizing these matrices. In contrast
to before we will find here, that $\bar{W}_k$ might depend on $U_1,
\ldots ,U_l$.

We use now Lemma \ref{Le1} to present a constructive method to
compute all local unitaries as functions of a few variables. If
some unitary, $U_i$ cannot be determined in this way, we write
$U_i=e^{-i\gamma_i Z_i} e^{-i\beta_i X_i} e^{-i\alpha_i Z_i}$ (up
to a phase). Then $\ket{\Psi}=\otimes_j U_j\ket{\Phi}$ iff
$e^{i\alpha_i Z_i} \bar{W}_i\ket{\Psi}=\otimes_{j\neq i}
U_j\ket{\Phi}$, where $\bar{W}_i=e^{i\beta_i X_i} e^{i\gamma_i
Z_i}$. That is, in this case we set $\bar{V}_i=\one$, $k_i=0$, and
$\bar{W}_i=e^{i\beta_i X_i} e^{i\gamma_i Z_i}$ in Eq (\ref{LU}). We
will say then that we consider $U_i$ as a variable.

The constructive method to compute now $\bar{V}_k$ and $\bar{W}_k$
in Eq (\ref{LU}) is as follows: (1) If there exists a system $i$
such that $\rho_i\not\propto \one $ compute $\bar{V}_i$,
$\bar{W}_i$ using that
$\bar{W}_i\rho_i\bar{W}_i^\dagger=\bar{V}_i\sigma_i\bar{V}_i^\dagger$
($k_i=0$). Furthermore, compute $\bar{V}_k$ and $\bar{W}_k(U_i)$
for any system $k$ with $\rho_{ik}\neq \rho_i\otimes \one $ using
Lemma \ref{Le1}. (2) For all systems $i$ for which $\rho_i \neq
\one$ apply the unitaries $\bar{W}_i$ ($\bar{V}_i$) to $\ket{\Psi}$
($\ket{\Phi}$) resp. and measure system $i$ in the computational
basis thereby reducing the number of systems (see Eq
(\ref{cond1})). After this step we have $\rho_i \propto \one$ $
\forall i$. Then we continue as follows: (3) Consider the two qubit
reduced states: (3a) There exist systems $i,j$ such that
$\rho_{ij}\neq \one $. W. l. o. g. we choose $i=1$, consider $U_1$
as variable, and set $\bar{V}_1=\one$, $k_1=0$ and
$\bar{W}_1=e^{i\beta_1 X_1} e^{i\gamma_1 Z_1}$. Then, compute
$\bar{V}_j$ and $\bar{W}_j(U_1)$ using Lemma \ref{Le1} for any
system $j$ with $\rho_{1j}\not\propto \one$. Let us denote by $J_2$
the set of systems for which $\rho_{1j}\not\propto \one$. (3b) If
there exists no system $i,j$ such that $\rho_{ij}\not\propto \one $
consider $U_1$ and $U_2$ as variables and set $\bar{V}_i=\one$,
$k_i=0$ and $\bar{W}_i=e^{i\beta_i X_i} e^{i\gamma_i Z_i}$, for
$i=1,2$. Furthermore, set $J_2=\{2\}$. (4) Consider the
three--qubit reduced states: (4a) If there exists a system $k$ such
that $\rho_{1 j k}\neq \rho_{1j}\otimes \one$ for some $j\in J_2$
compute $\bar{V}_k$ and $\bar{W}_k(U_1,U_j)$ using Lemma \ref{Le1}.
Determine for any system $k$ with $\rho_{1 j k}\neq
\rho_{1j}\otimes \one$ $\bar{V}_k$ and $\bar{W}_k(U_1,U_j)$ (if
they are not already determined). (4b) If there exists no system
$k$ such that $\rho_{1jk}\neq \rho_{1j}\otimes \one$ include $U_3$
as variable. (5) Continue in this way until all unitaries are
either determined as functions of a few unitaries, or are free
parameters. If at some point it is not possible to choose
$\bar{V}_k$ or $\bar{W}_k$ unitary, e.g. if the eigenvalues of the
operators occurring in Eq (\ref{Yi}) do not coincide, the states
are not LU--equivalent.

Once all unitaries, $\bar{V}_i$ are determined and all unitaries
$\bar{W}_i$ are determined as functions of a few variables, we have
that $\ket{\Psi}\simeq_{LU} \ket{\Phi}$ iff there exists a
bitstring ${\bf k}$ and phases $\{\alpha_i\}_{i=0}^n$, such that Eq
(\ref{LU}) is fulfilled. In order to check the existence of the
local phase gates in Eq. (\ref{LU}) (for some bitstring ${\bf k}$),
we use Lemma \ref{LemmaPhase}. It is important to note here that
the state on the right hand side of Eq. (\ref{LU}) is completely
determined, thus, the set $K_\Psi$ in Lemma \ref{LemmaPhase} can be
determined and therefore this lemma can be applied. The states are
LU--equivalent iff the conditions in Lemma \ref{LemmaPhase} are
fulfilled for some bitstring ${\bf k}$. Note that the unitaries,
$U_i$ which transform $\ket{\Phi}$ into $\ket{\Psi}$ are then given
by $U_i=\bar{W}^\dagger_i U(\alpha_i) X^{k_i}\bar{V}_i$ (up to a
global phase) \footnote{Note that the phases $\alpha_i$ can be
easily computed.}. These unitaries are uniquely determined up to
the symmetry of the state.

Note that a pure state has the property that
$\rho_{i_1,\ldots,i_l,k}=\rho_{i_1,\ldots,i_l}\otimes \one_k$ iff
for any outcome of any Von Neumann measurement on systems
$i_1,\ldots,i_l$, system $k$ is maximally entangled with the
remaining systems. Only in this case we have to add another unitary
as a variable. It is clear that two states, $\ket{\Psi},\ket{\Phi}$
with $\rho_{i_1,\ldots,i_l,k}=\rho_{i_1,\ldots,i_l}\otimes \one$
and $\sigma_{i_1,\ldots,i_l,k}\neq\sigma_{i_1,\ldots,i_l}\otimes
\one$ can neither be LU--equivalent nor posses the same
entanglement. Thus, the method presented above suggests that in
order to characterize the non--local properties of multipartite
states, one should first identify the class (as described above) to
which the state belongs to and then determine within this class the
entanglement of the state. It might well be, that the different
classes lead to different applications. For instance, the states
used for error correction, one way quantum computing and quantum
secret sharing have the property that all single qubit reduced
states are completely mixed.

Before we consider now some examples, let us mention that the worst
case, i.e. the case which involves the largest number of variables,
is the one where the reduced state of any bipartite splitting of
$\lceil n/2 \rceil$ systems versus the rest are maximally mixed. In
this case we have $\lceil n/2 \rceil$ unitaries as variables. Note
however, that there are very few instances, where those states do
exist \cite{Bra03}.

In order to illustrate the power of this method we consider first
the simplest examples of two and three qubit states. The standard
form of a two qubit state is $
\ket{\Psi}=\lambda_1\ket{00}+\lambda_2\ket{11}$. Thus, the method
above tells us that if $\lambda_1\neq \lambda_2$, i.e. $\rho_i\neq
\one$, then, $\ket{\Psi}\simeq_{LU}\ket{\Phi}$ iff the Schmidt
coefficients $\lambda_i$ are the same. For $\lambda_1=\lambda_2$ it
is straightforward to show that the unitaries, $U_i$, which are
obtained using the method above for the states $\ket{\Phi^+}\equiv
\ket{00}+\ket{11}$ and some LU--equivalent state
$V_1V_2\ket{\Phi^+}$ are $U_1=V_1W$ and $U_2=V_2W^\ast$ for any
unitary $W$. The reason why the unitaries $U_i$ are not completely
determined by $V_i$ is due to the symmetry of the state,
$\ket{\Phi^+}=W\otimes W^\ast \ket{\Phi^+}$ $\forall W$ unitary.

For three qubits the method is almost equally simple. First, we
transform both states into their trace decomposition. If non of the
reduced states is completely mixed, we simply compare their
standard forms (Theorem $1$). If there exists some $i$ such that
$\rho_i\neq\one$, we know that $U_i=U(\alpha_i)$. We measure system
$i$ in the computational basis and are left with two two--qubit
states (see Eq. (\ref{cond1})). In case those states are
LU--equivalent we apply the corresponding unitaries and use Lemma
\ref{LemmaPhase} to find out whether the three qubit states are
LU--equivalent or not. For the remaining case, where $\rho_i=\one$
$\forall i$ it can be easily shown that $\ket{\Psi}$ is
LU--equivalent to the GHZ--state,
$\ket{\Psi_0}=\ket{000}+\ket{111}$ \cite{Kr09b}. Even without using
this fact it can be easily shown that also in this case the method
presented above leads directly to the right unitaries (up to the
symmetry of the states) for two states which are LU--equivalent
(for details see \cite{Kr09b}).

With the same method the LU--equivalence classes of up to
$5$--qubit states are investigated in \cite{Kr09b}. We will show
there, for instance, that for $4$--qubit states with
$\rho_{ij}=\one$ for some $i,j$ (which is the hardest class of
states using the method presented above), the LU--equivalence class
is determined by only three parameters. Thus, also the entanglement
of those states is completely determined by the fact that system
$ij$ is maximally entangled to the other two qubits and those three
parameters, to which also an operational meaning will be given
\cite{Kr09b}. This example shows already that the method presented
here does not only give necessary and sufficient conditions for the
LU--equivalence of arbitrary multipartite states, but also leads to
a new insight into their entanglement properties.

Finally, let us note that the results presented above serve also as
a criterion of LU--equivalence for certain mixed and also
$d$--level states. For instance, if there exists at least one
non--degenerate eigenvalue of $\rho$ ($\sigma$) with corresponding
eigenvectors $\ket{\Psi}$ ($\ket{\Phi}$) resp., then
$\rho\simeq_{LU}\sigma$ implies that $\ket{\Psi}\simeq_{LU}
\ket{\Phi}$. Using the method presented here, all the unitaries,
which transform $\ket{\Psi}$ into $\ket{\Phi}$ can be determined
and therefore it is straightforward to check if one of them also
converts $\rho$ into $\sigma$.

In summary, a systematic way to show the LU--equivalence of
arbitrary multipartite pure states is presented. The results
derived here also lead to a new insight into the entanglement
properties of the multipartite states. Studying the different
classes specified here, allows one to identify new parameters
characterizing entanglement \cite{Kr09b}. In particular, for
generic states all the parameters occurring in the standard form
determine, like in the bipartite case, the entanglement contained
in the state.

The author would like to thank Hans Briegel for continuous support
and interest in this work and acknowledges support of the FWF
(Elise Richter Program).


\section{Appendix: Interconvertability by local phase gates}

In order to prove Lemma \ref{LemmaPhase} we will make use of the
following lemma, where we use the same notation as before.

\begin{lemma}
$\ket{\Psi}$ can be converted into $\ket{\Phi}$ by local unitary
phase gates iff there exist phases $\{\bar{\alpha}_i\}_{i=0}^{n}$
such that $\ket{\Psi_{\bf 0}}$ is converted into
$\ket{\Phi_{\bar{\alpha}_i}}$ by local unitary phase
gates.\end{lemma}

\begin{proof}

If $\ket{\Psi}=e^{i\alpha_0}\bigotimes_{i=1}^n U_i(\alpha_i)
\ket{\Phi}$ then choosing $\bar{\alpha}_i=\alpha_i$ for $i\in
\{0,\ldots, n\}$ fulfills the condition. To prove the inverse
direction we assume that there exist phases
$\{\bar{\alpha}_i\}_{i=0}^{n}$ such that $\ket{\Psi_{\bf
0}}=e^{i\alpha_0}\bigotimes_{i=1}^n U_i(\alpha_i)
\ket{\Phi_{\bar{\alpha}_i}}$ for some phases $\{\alpha_i\}$.
Defining the projector $P=\sum_{{\bf k}\not\in K} \proj{{\bf k}}$
we have $P\ket{\Psi_{\bf 0}}=\ket{\Psi}$ and
$Pe^{i\alpha_0}\bigotimes_{i=1}^n U_i(\alpha_i)
\ket{\Phi_{\bar{\alpha}_i}}=e^{i\alpha_0}\bigotimes_{i=1}^n
U_i(\alpha_i) P\ket{\Phi_{\bar{\alpha}_i}}$ and therefore
$\ket{\Psi}=e^{i\alpha_0}\bigotimes_{i=1}^n U_i
(\alpha_i)\ket{\Phi}$.

\end{proof}

Let us now use the lemma above to prove Lemma \ref{LemmaPhase}.
\begin{proof}

Due to the lemma above it remains to show that for any state
$\ket{\psi}$ with $\bra{{\bf k}}\psi\rangle\neq 0$ $\forall {\bf
k}$ we have that $\ket{\psi}=e^{i\alpha_0}\bigotimes_{i=1}^n U_i
(\alpha_i) \ket{\phi}$  iff condition (i) and (ii) in Lemma
\ref{LemmaPhase} are satisfied. Note that Eq. (\ref{cond1}) is
equivalent to $\bra{0k}\psi\rangle \bra{1l}\psi\rangle
\bra{1k}\phi\rangle \bra{0l}\phi\rangle= \bra{1k}\psi\rangle
\bra{0l}\psi\rangle \bra{0k}\phi\rangle \bra{1l}\phi\rangle, $
where $0,1$ is acting on system $i$ and $k,l$ denote the
computational basis states of the remaining $n-1$ qubits.

Let us now prove the {\it only if} part: If
$\ket{\psi}=e^{i\alpha_0}\bigotimes_{i=1}^n U_i \ket{\phi}$ then
$\bra{{\bf i}}\psi\rangle=e^{i\phi_{{\bf i}}}\bra{{\bf
i}}\phi\rangle$, with $\phi_{{\bf i}}=\alpha_0+\sum_k \alpha_k
i_k$, which implies (i). Condition (ii) (for $i=1$) is then
equivalent to $e^{i(\phi_{0k}+\phi_{1l})}
x_{kl}=e^{i(\phi_{1k}+\phi_{0l})} x_{kl},$  where
$x_{kl}=\bra{0k}\phi\rangle \bra{1l}\phi\rangle \bra{1k}\phi\rangle
\bra{0l}\phi\rangle$. It is easy to see that this condition is
fulfilled since $e^{i(\phi_{0k}-\phi_{1k})}=e^{-i\alpha_1}$ $
\forall k$. In the same way one can show that the conditions for
$i\neq 1$ are fulfilled.

{\it If}: Condition (i) implies that $\bra{{\bf i}}\Psi\rangle=e^{i
\phi_{\bf i}}\bra{{\bf i}}\Phi\rangle$, for some phases $\phi_{\bf
i}$. Condition (ii) (for $i=1$) implies then that
$e^{i(\phi_{0k}-\phi_{1k})} =e^{i(\phi_{0l}-\phi_{1l})}$ $\forall
k,l$,  since $x_{kl}=\bra{0k}\phi\rangle \bra{1l}\phi\rangle
\bra{1k}\phi\rangle \bra{0l}\phi\rangle \neq 0$ $\forall k,l$.
Thus, $e^{i(\phi_{0k}-\phi_{1k})}$ must be independent of $k$ and
therefore, we have $e^{i(\phi_{0k}-\phi_{1k})}=e^{-i\alpha_1}$, or
equivalently, $e^{i\phi_{k_1,k}}=e^{i
(\alpha_1^{(k_1)}+\phi_{1k})}$, where $\alpha_1^{(0)}=-\alpha_1$
and $\alpha_1^{(1)}=0$. Similarly we have
$e^{i(\phi_{k_10k_3,\ldots, k_n}-\phi_{k_11k_3\ldots,
k_n})}=e^{-i\alpha_2}$ and therefore
$e^{i\phi_{k_1,k_2,k_3\ldots,k_n}}=e^{i
(\alpha_1^{(k_1)}+\alpha_2^{(k_2)}+\phi_{11k_3,\ldots,k_n})}$.
Continuing in this way we find $e^{i\phi_{k_1,\ldots
k_n}}=e^{i\alpha_0}e^{i\sum_j \alpha_j k_j}$, where
$\alpha_0=\phi_{1\ldots 1}-\sum \alpha_i$. Thus, we have
$\ket{\psi}=e^{i\alpha_0}\bigotimes_{i=1}^n U_i (\alpha_i)
\ket{\phi}$ with $U_i(\alpha_i)=\mbox{diag}(1,e^{i\alpha_i})$.
Using the lemma above, this implies that $\ket{\Psi}=e^{i\alpha_0}
\bigotimes_{i=1}^n U_i(\alpha_i) \ket{\Phi}$.

\end{proof}

\end{document}